%% file: main.tex
\documentclass[a4paper,USenglish,cleveref,autoref,thm-restate]{lipics-v2021}
\usepackage{xspace}
\usepackage{mathtools}
\usepackage{interval}

\hideLIPIcs 
\nolinenumbers 

\EventEditors{Amit Chakrabarti and Chaitanya Swamy}
\EventNoEds{2}
\EventLongTitle{Approximation, Randomization, and Combinatorial Optimization. Algorithms and Techniques (APPROX/RANDOM 2022)}
\EventShortTitle{\mbox{\scriptsize APPROX/RANDOM 2022}}
\EventAcronym{APPROX/RANDOM}
\EventYear{2022}
\EventDate{September 19--21, 2022}
\EventLocation{University of Illinois, Urbana-Champaign, USA (Virtual Conference)}
\EventLogo{}
\SeriesVolume{245}
\ArticleNo{45}

\newtheorem{algorithm}{Algorithm}
\newtheorem{linprog}{Linear Program}
\theoremstyle{claimstyle}
\newtheorem{empiricalclaim}[theorem]{Empirical Claim}

\intervalconfig{soft open fences}
\DeclarePairedDelimiter{\ceil}{\lceil}{\rceil}

\DeclarePairedDelimiter\tuple{\langle}{\rangle}
\DeclarePairedDelimiter\upto{\lbrack}{\rbrack}
\DeclarePairedDelimiter\fl{\lfloor}{\rfloor}

\DeclarePairedDelimiterX\set[1]{\lbrace}{\rbrace}{
    
    #1%
}

\newcommand*{\Nat}{\mathbb{N}}
\newcommand*{\Real}{\mathbb{R}}
\renewcommand*{\leq}{\leqslant}
\renewcommand*{\geq}{\geqslant}
\newcommand{\ALGTS}{\textsc{ATS}\xspace}
\newcommand{\OPT}{\textsc{Opt}\xspace}
\newcommand{\tn}{\tilde{n}}
\newcommand{\wait}{w}
\newcommand{\waitTS}{w^\mathrm{ATS}}

\newcommand{\waitinit}{w^\mathrm{init}}
\newcommand{\X}{\mathcal{X}}
\newcommand{\Y}{\mathcal{Y}}
\newcommand{\dist}{\mathrm{dist}}
\newcommand{\fc}{\mathrm{open}}
\newcommand{\fcg}{\mathrm{open}(f)}
\newcommand{\sensa}{\lambda}
\newcommand{\sensb}{\xi}

\newcommand{\lpone}{LP1\xspace}
\newcommand{\lptwo}{LP2\xspace}

\title{Online Facility Location with Linear Delay}

\author{Marcin Bienkowski}
{Institute of Computer Science, University of Wrocław, Poland}
{marcin.bienkowski@cs.uni.wroc.pl}
{https://orcid.org/0000-0002-2453-7772}{Supported by Polish National Science Centre grant 2016/22/E/ST6/00499.}
\author{Martin Böhm}
{Institute of Computer Science, University of Wrocław, Poland}
{boehm@cs.uni.wroc.pl}
{https://orcid.org/0000-0003-4796-7422}{}
\author{Jarosław Byrka}
{Institute of Computer Science, University of Wrocław, Poland}
{jaroslaw.byrka@cs.uni.wroc.pl}
{https://orcid.org/0000-0002-3387-0913}{Supported by Polish National Science Centre grant 2020/39/B/ST6/01641.}
\author{Jan Marcinkowski}
{Institute of Computer Science, University of Wrocław, Poland}
{jasiekmarc@cs.uni.wroc.pl}
{https://orcid.org/0000-0002-6517-0014}{}
\authorrunning{M. Bienkowski, M. Böhm, J. Byrka, and J. Marcinkowski}
\Copyright{Marcin Bienkowski, Martin Böhm, Jarosław Byrka, and Jan Marcinkowski}

\supplementdetails[]{Software}{https://github.com/bohm/fl-double-sided-waiting}

\begin{document}

\maketitle

\begin{abstract}
In the problem of online facility location with delay, a
sequence of $n$ clients appear in the metric space, and they need to be
eventually connected to some open facility. The clients do not have to be
connected immediately, but such a choice comes with a~certain penalty: each
client incurs a waiting cost (equal to the difference between its arrival and
its connection time). At any point in time, an algorithm may decide to open a
facility and connect any subset of clients to it. That is, an algorithm needs to
balance three types of costs: cost of opening facilities, costs of connecting
clients, and the waiting costs of clients.
We study a natural variant of this problem, where clients may be connected also to an \emph{already open} facility,
but such action incurs an extra cost: an algorithm pays for waiting of the
facility (a~cost incurred separately for each such ``late'' connection). This is
reminiscent of online matching with delays, where both sides of the connection
incur a waiting cost. We call this variant \emph{two-sided delay} to
differentiate it from the previously studied \emph{one-sided delay}, where
clients may connect to a~facility only at its opening time.

We present an $O(1)$-competitive deterministic algorithm for the two-sided delay
variant. Our approach is an extension of the approach used by Jain, Mahdian and
Saberi [STOC 2002] for analyzing the performance of \emph{offline}
algorithms for facility location. To this end, we substantially
simplify the part of the original argument in which a bound on the
sequence of factor-revealing LPs is derived. We then show how to
transform our $O(1)$-competitive algorithm for the two-sided delay
variant to $O(\log n / \log \log n)$-competitive deterministic
algorithm for one-sided delays. This improves the known $O(\log n)$
bound by Azar and Touitou [FOCS 2020]. We note that all previous
online algorithms for problems with delays in general metrics have at
least logarithmic ratios.
\end{abstract}
\input{chapters/intro}
\input{chapters/algo}

\input{chapters/factrevlp}
\input{chapters/computational}
\input{chapters/onesided}

\bibliography{references}

\end{document}

%% file: chapters/intro.tex

\section{Introduction}

The facility location problem~\cite{AaByMa16} is one of the best-known examples
of network design problems, extensively studied both in operations research and
in computer science. The problem is defined in a metric space $\X$.
An algorithm is given a set of $n$ clients and its
goal is to open a~set of facilities (chosen points of $\X$) minimizing the total
cost, defined as the sum of costs of opening facilities plus the costs of connecting
clients. Our focus is on the \emph{non-uniform} case, where the opening cost of
a facility may depend on its position in $\X$. The connection cost of a given
client is simply its distance to the nearest open facility. This simple statement
hides a~surprisingly rich combinatorial structure and gave rise to series of
algorithms and extensions. In particular, the problem is NP-complete and
APX-hard~\cite{GuhKhu99} and its approximation ratio has been studied in a long
sequence of
improvements~\cite{ShTaAa97,KoPlRa00,JaMaSa02,JaMMSV03,ChuShm03,ChaGuh05,MaYeZh06,ByrAar10},
with the current record of 1.488 proved by Li~\cite{Li13}.

In the online scenario, the set of clients is not known up-front, but it is
revealed to an~(online) algorithm one element at a time. Once a client becomes
known, an algorithm has to make an~irrevocable and immediate decision whether to
open additional facilities and to which facility the current client should be
connected.\footnote{The best option is clearly to connect a given client to the
closest facility, but some naturally defined algorithms may not have this
property.} As in the offline scenario, the algorithm is compared to the best
\emph{offline} solution, and in online scenarios, we use the name
\emph{competitive ratio} instead of approximation ratio. This scenario has been
fully resolved: tight asymptotic lower and upper bounds of $\Theta(\log n / \log
\log n)$ are known both for randomized and deterministic
algorithms~\cite{Meyers01,AnBeUH04,Fotaki07,Fotaki08}.

In the last few years, many online problems have been considered in scenarios
\emph{with delays}. In the case of online facility location, first studied by
Azar and Touitou~\cite{AzaTou19}, the clients arrive in time, and while each of
them has to be connected eventually, such action does not have to be executed
immediately. This additional degree of freedom comes, however, with a price:
each waiting client incurs an extra cost that (in the basic setting studied in
this paper) is equal to its total waiting time (time between its
arrival and its connection). We note that in terms of achievable competitive
ratios, the classic models and models with delays are rarely comparable as the
possibility of delaying actions is allowed also in the benchmark offline solution.


\paragraph*{Facility location with one-sided delay} 

This variant has been introduced and studied in~\cite{AzaTou19,AzaTou20}. There
each facility is \emph{ephemeral}: it is opened only momentarily at time $t$
chosen by an~algorithm, and all connections to this facility must be made at
time $t$. The algorithm can open another facility at the same location
at a different time $t'$, but the opening cost must be paid again. The waiting
costs of clients are as described above.

The best known algorithm for this problem variant is $O(\log
n)$-competitive~\cite{AzaTou20}.  (Interestingly, it is not known whether this
particular variant admits constant-factor approximation in the offline setting
when all client arrivals are known up-front.)

\paragraph*{Facility location with two-sided delay} 

We propose the following slight deviation from the one-sided variant described
above: once a~facility becomes open at time $t$, it remains open forever and can
be connected to in the future. However, any client that connects to such
facility at time $t' > t$ needs to pay an~\emph{additional} waiting cost of $t'
- t$. We call this amount \emph{facility-side waiting cost}, which needs to be
paid on top of the ``standard'' \emph{(client-side) waiting cost}. We emphasize that
such connections, dubbed \emph{late connections}, can be made both by clients
that arrived before facility opening and also after this time. Similarly to the
one-sided delay model, we allow opening of multiple facilities in the same location
at different points in time.


The two-sided model can be seen as an approximation of the
\emph{early-late adopter behavior} in crowdfunding models. In
crowdfunding platforms for technology
products~\cite{stanko2016crowdfunding} such as Kickstarter, any
specific project is started by gathering initial contributions by
enthusiasts up to a certain threshold, which should (in an ideal case)
imply opening of a production line for the specified technology
product. Contributors in this pre-production phase are called
\emph{early adopters}.

As we move forward in time, while the production may already
begin, it may still be possible for so-called \emph{late adopters} to
join the crowdfunding project on the Kickstarter website and receive
the final product. As early adoption is more beneficial for the
producer of the technology product, late adoption is sometimes penalized
with an increased cost of the same product compared to the early adoption.

In the framework of crowdfunding models, we can see online
facility location with two-sided delay as facility location in an
early-late adopter setting, where a technology product can be
manufactured at multiple factories and late adopters may join into the
crowdfunding scheme and thus contribute towards offsetting the cost of
the production while it is in progress. However, the clients who join
late need to deal with the missed-opportunity cost (the client-side
cost) as well as the late adopter increase in price (the facility-side
waiting cost).

\subsection{Our Results and Techniques}

Our first positive contribution is showing that facility
location with two-sided delay admits a constant-competitive online
algorithm, which is an important open problem for the one-sided
case. Namely, we show:

\begin{theorem}
\label{thm:two-sided}
There exists an $3.869$-competitive deterministic algorithm for the online facility location
problem with two-sided delay, where all waiting costs are equal to the waiting
times. 
\end{theorem}

We analyze a natural greedy algorithm, which grows budgets with increasing
waiting delays and opens facilities for subsets of clients once sums of these
budgets reach certain thresholds. To analyze this algorithm, we use dual fitting
methods. Our analysis is a substantial extension of the approach used by Jain et
al.~\cite{JaMaSa02,JaMMSV03} for analyzing the performance of \emph{offline}
algorithms for (non-delayed) facility location. 

The central part of the analysis is a linear program (LP), parameterized by an integer
$k$, whose objective value is an upper bound on the competitive ratio of our
algorithm \emph{provided the number of clients in the input is at most
$k$}.
As the objective function of this LP grows with $k$, simply solving the LP would
yield the correct upper bound only for instances of limited size. As a replacement for the technical original argument in~\cite{JaMMSV03} we propose a much more intuitive one that is based on an upper bound to this sequence by the value of a finite linear program.

We would like to stress that we see our novel approach to the
competitive analysis of facility location dual-fitting LP as an
important contribution of this paper to the area of facility location
with delays. Our approach has a significant computer-assisted
component (\Cref{sec:computational}) and it can thus be quickly
deployed to give provable estimates on the potential competitive ratio
or an approximation ratio of factor-revealing linear programs in other
settings. Our code for the algorithmic part is publicly
available~\cite{datarepository}.

\medskip

Our second result is showing that $O(1)$-competitive algorithm for 
the two-sided variant yields an improved guarantee also for the
one-sided variant. 

\begin{theorem}
\label{thm:one-sided}
There exists an $O(\log n / \log \log n)$-competitive deterministic algorithm for the 
online facility location problem with one-sided delay,
where all waiting costs are equal to the waiting times. 
\end{theorem}

We prove this result via a reduction that can be applied to any
algorithm solving the two-sided variant provided it satisfies 
a certain technical condition that we call \emph{sensibility}.
Informally speaking this property means that the waiting costs associated
with late connections are not very large; 
we defer the precise definition to \cref{sec:factor-revealing}.

\cref{thm:one-sided} improves the known $O(\log n)$-bound by Azar and
Touitou~\cite{AzaTou20}. While the improvement is small, we note that all
previous online algorithms for problems with delays have at
least logarithmic ratios (in the number of used points of the metric space), so ours is the first to break this natural barrier. 

\subsection{Related work}

Recently many online graph problems have been considered in a variant that
allows requests to be delayed. Apart from the facility location problem studied in this
paper, examples include the Steiner tree problem~\cite{AzaTou19,AzaTou20}, multi-level
aggregation~\cite{BBCJSS13,Chroba14,BuFeNT17,AzaTou19,BBBCDF20,BBBCDF21},
Steiner forest/network~\cite{AzaTou20}, directed Steiner tree~\cite{AzaTou20},
multi-cut~\cite{AzaTou20}, online
matching~\cite{EmKuWa16,AACCGK17,AzChKa17,BiKrSc17,BiKrLS18,LiPaWW18,EmShWa19,AzaFan20,AzReVa21},
set cover~\cite{CaPrSV18,AzChKT20} and $k$-server (known in this setting as online service
with delay)~\cite{AzGaGP17,BiKrSc18,AzaTou19}.

That said, the concept of delaying requests itself is not new. Famous studied problems
include the TCP acknowledgement
problem~\cite{DoGoSc01,KaKeRa03} and joint replenishment
problem~\cite{BuKLMS13,BBCJNS14} (that are equivalent to the recently studied
multi-level aggregation problem on one-level or two-level trees).

\paragraph*{General waiting costs}

The waiting costs considered in this paper are equal to waiting times.
Another studied case are deadlines, where waiting costs are zero till a
request-specific time (the deadline) and infinite afterwards. For some problems, easier
algorithms or better bounds are known when the waiting costs are in the deadline
form (see, e.g.,~\cite{BBCJNS14,BuFeNT17}).

Many of the results listed above can be extended to waiting costs being
arbitrary non-decreasing left-continuous functions of waiting times. These extensions
are straightforward if an algorithm is defined by simple thresholds on (sums of)
waiting costs; when these thresholds are reached, they trigger an appropriate
action of the algorithm. For instance, algorithms for the TCP
acknowledgement problem were constructed for linear waiting costs, but they can be
trivially extended to general costs.

There are however a few cases where general waiting costs are more problematic. 
Most notably, the online matching problem was studied for linear waiting 
costs~\cite{EmKuWa16,AACCGK17,AzChKa17,BiKrSc17,BiKrLS18,EmShWa19,AzaFan20},
then shown to be more difficult (in terms of achievable competitive ratios)
for convex waiting costs~\cite{LiPaWW18}, and only recently competitive algorithms 
were shown for concave waiting costs~\cite{AzReVa21}.

The algorithms presented in this paper also fall into the latter category: 
our LP-based analysis heavily depends on the linearity of the
waiting functions, and thus our algorithms cannot be easily extended 
to general waiting costs. (We note that the $O(\log n)$-competitive algorithm 
by Azar and Touitou~\cite{AzaTou20} can handle arbitrary waiting costs.)

\subsection{Remark about Facility-Side Waiting Costs}

Recall that in the two-sided variant that we study in our paper, 
we assume that each late connection incurs an additional waiting cost 
at the facility side (equal to the time that passes between facility
opening and client connection). Below, we argue that setting this cost
to zero would cause the optimal competitive ratio to be 
$\Theta(\log n / \log \log n)$. This shows that for the $O(1)$-competitive 
result of \cref{thm:two-sided}, some assumptions about waiting cost functions 
are necessary. 

\begin{observation}
In the two-sided variant of the facility location problem with delays, 
setting facility-side waiting costs to zero causes the optimal competitive 
ratio (both deterministic and randomized) to become asymptotically equal 
to $\Theta(\log n /\log \log n)$.
\end{observation}


\begin{proof}
Assume no penalty for facility-side waiting. For any input instance, without
an~increase of the cost, \OPT may open all its facilities at time $0$ and connect
all clients immediately when they arrive. Thus, the optimal solution incurs no
waiting cost at all.

As the waiting cost can be avoided also for an online algorithm (by serving all
clients immediately upon their arrival), the desired competitive ratio can be
attained by running an~$O(\log n / \log \log n)$-competitive algorithm
(deterministic or randomized) for the online facility location problem (without
delays)~\cite{Meyers01,Fotaki08},

For showing a lower bound, we may simply use an adversarial strategy for
the online facility location problem (without delays)~\cite{Fotaki08}; however,
now the next request is presented only after an algorithm serves the
previous one. This way, waiting becomes useless for an~online algorithm, and the
lower bound of $\Omega(\log n / \log \log n)$~\cite{Fotaki08} applies also for
the waiting model.
\end{proof}

\subsection{Preliminaries}

We use the following notions throughout the paper.

The pair $(\X,\dist)$ denotes the underlying metric space with its distance function,
and $\Y \subseteq \X$~denotes the positions of potential facilities.
The function $\fc: \Y \to \Real_{\geq 0} \cup \{\infty\}$ defines the cost of 
opening a facility at a specific location.

The clients are numbered from $1$ to $n$ and arrive in time; each is associated with a point in $\X$. Their total number is not known up-front to 
an~online algorithm.
For a client $j$, we use $x_j$ to denote its position and $t_j$ to denote the time of its arrival.
We say that a client is \emph{active} from the time of its
arrival until it gets connected to a facility by an online algorithm, and it is \emph{inactive}
after the connection.

At any time, an algorithm may open a facility at any point $y \in \Y$, 
paying \emph{opening cost}~$\fc(y)$. 
An~active client $j$ may be connected by an algorithm at time $t^c_j \geq t_j$
to a facility that was open at time $\tau$, such that
\begin{itemize}
    \item $\tau = t^c_j$, for the one-sided delay variant;
    \item $\tau \leq t^c_j$, for the two-sided delay variant.
\end{itemize}
Such a connection incurs a \emph{connection cost} $\dist(x_j, y)$ 
and two types of waiting costs: a~\emph{client-side waiting cost} $t^c_j - t_j$ and 
a \emph{facility-side waiting cost} $t^c_j - \tau$. Note that the latter waiting cost is 
always zero for the one-sided delay variant.

The goal of an algorithm is to minimize the total cost, defined as the sum of 
opening costs, connection costs, and waiting costs.

%% file: chapters/algo.tex
\section{Algorithm for the Two-Sided Variant}

We will now describe the algorithm for the online facility location with
two-sided linear waiting cost. 
Our algorithm is parameterized with a constant $\gamma > 1$ that will be fixed later.
We say that an~active client $j$ at time $t \geq t_j$, after experiencing waiting cost of $t -
t_j$, has a~\emph{connectivity budget} 
\[
    \alpha_j(t)= \gamma \cdot (t - t_j).
\]
In the analysis, we use $\alpha_j$ to denote the connectivity budget of client $j$ at
time $t^c_j$, when it becomes connected to a facility.

At any time $t$, for each potential facility location $y \in \Y$, an active client~$j$ has an offer of
a~contribution towards the opening of a facility at $y$ equal to
$\beta_{j}^t(y) = \max \{ 0, \alpha_j(t) - \dist(x_j, y)\}$. When client $j$ becomes
connected (and inactive), it no longer offers any contribution.

The algorithm follows the natural continuous flow of time of the online sequence
and reacts to the events occurring throughout its runtime.
\begin{algorithm}
  \label{algo:two-sided}
    At any time $t$, do the following.
    \begin{alphaenumerate}
        \item If a new client $j$ arrives at point $x_j$ (client $j$
        becomes active): From this time onward start growing its
        connectivity budget.
        
        \item \label{item:case2} If for any location $y \in \Y$, the sum of offered contributions
        $\beta_{j}^t(y)$ towards this location from all the currently active
        clients reaches $\fc(y)$ (the facility opening cost at $y$): 
        Let $A^t(y)$ be the set of active clients $j$ for which
        $\alpha_j(t) \geq \dist(x_j, y)$, i.e., those that can afford the
        distance. Open a facility at $y$, and connect every
        client $j \in A^t(y)$ to it. 
        
        \item \label{item:case3} If for a facility that has already been open at time $\tau \leq t$ at 
        location $y$, and for an active client~$j$, it holds that $t -
        \tau = \alpha_j(t) - \dist(x_j, y)$: Connect client $j$ to this facility. We call this 
        action \emph{late connection}.
    \end{alphaenumerate}
In Case \ref{item:case2} and Case \ref{item:case3}, all clients that get connected become inactive.
\end{algorithm}

\subsection{Basic observations}

One may observe that \cref{algo:two-sided} is a generalization of 
(the simpler version of) the algorithm by Jain et al.~\cite{JaMaSa02}. Furthermore, if it
is run on an instance where all clients appear at the same time, it produces the
same solution (the same facility locations and the same connections) as the
original offline approximation algorithm. 

It is convenient to think that the connectivity budget of a client is first spent on connection cost,
and the remaining part either contributes to the opening of a new facility or pays for the facility-side 
waiting cost of an already open facility.

\begin{observation}
The total cost of the solution produced by \cref{algo:two-sided} is 
$(1+1/\gamma) \cdot \sum_j \alpha_j$.
\end{observation}

\begin{proof}
The sum of opening costs, connection costs and facility-side waiting costs
in the produced solution equals the sum of final connectivity budgets 
$\sum_j \alpha_j$. The total client-side waiting
cost is $(1/\gamma) \cdot \sum_j \alpha_j$. 
\end{proof}

To estimate the competitive ratio of the algorithm, it thus suffices to
compare the cost of the optimal solution to $(1 + 1/\gamma) \cdot \sum_j \alpha_j$, which we do 
in \cref{sec:factor-revealing}.

\subsection{Sensibility}

Our later construction for the one-sided waiting variant requires that the used
online algorithm for the two-sided variant has the following property, which bounds the number of 
clients that connect late to an already open facility. 

\begin{definition}[sensibility]
\label{def:sensibility}
Fix any $\sensa > 1$ and $\sensb > 1$. An algorithm solving the
two-sided variant is called $(\sensa,\sensb)$-sensible if it satisfies the
following property for any facility $f$ opened at time $\tau$ and location $y$:
for any $w > 0$, the number of clients connected to $f$ within the interval
$\interval[open left]{\tau + w}{\tau + \sensa \cdot w}$ is at most $\sensb
\cdot \fc(y) / w$.
\end{definition}


We show that for some constants $\sensa$ and $\sensb$, our algorithm 
is $(\sensa,\sensb)$-sensible: if clients connecting late to an 
open facility incurred large (facility-side) waiting costs, then our algorithm would 
rather create a new copy of this facility, and connect these clients
to the copy.

\begin{lemma}
\label{lem:alg1_is_sensible}
Fix a parameter $\gamma > 1$ of \cref{algo:two-sided}. 
For any $\sensa \in (1, 1+1/(\gamma-1))$, \cref{algo:two-sided} is 
$(\sensa, (\gamma-(\gamma-1)\cdot \sensa)^{-1})$-sensible.
\end{lemma}

\begin{proof}
Fix a facility $f$ opened at time $\tau$ at location $y$.
Let $C_w$ be the set of clients that are connected late to $f$ 
within interval $\interval[open left]{\tau + w}{\tau + \sensa \cdot w}$.
Take any client $j \in C_w$ and let $t^c_j = \tau + h$ be its connection time.
We define the \emph{residual budget} of client $j$ at time $t$ as $r_j(t) = \alpha_j(t) - \dist(x_j, y)$. 

We now estimate the residual budget of client $j$ at time $\tau + w$.
Let $\tau + g$ be the time when its residual budget becomes zero, i.e., 
$r_j(\tau + g) = 0$. Its residual budget at time $\tau + h$ is then
$r_j(\tau + h) = \alpha_j(\tau+h) - \alpha_j(\tau+g) = \gamma \cdot (h-g)$.
On the other hand, $r_j(\tau + h) = (\tau + h) - \tau = h$, as $j$ forms a late connection to $f$ at time $\tau+h$.
Hence, $\gamma \cdot g = (\gamma - 1) \cdot h \leq (\gamma -1) \cdot \sensa \cdot w$,
which implies $r_j(\tau + w) = \gamma \cdot (w-g) \geq
(\gamma - (\gamma-1) \cdot \lambda) \cdot w$. 
(Note that for our choice of $\lambda$, this amount is positive.)

By the definition of \cref{algo:two-sided}, the residual budgets are spent either 
on opening new facilities or on facility-side waiting of an already opened facility.
Hence, at time $\tau + w$, the sum of residual budgets 
of all clients from $C_w$ cannot be larger than $\fc(y)$, as in such case \cref{algo:two-sided}
would open another copy of the facility at $y$.
This argument implies that $|C_w| \cdot (\gamma - (\gamma-1) \cdot \lambda) \cdot w \leq \fc(y)$,
which concludes the proof.
\end{proof}

For example, by \cref{lem:alg1_is_sensible}, \cref{algo:two-sided} 
with $\gamma = 2$ is 
$(3/2, 2)$-sensible.

%% file: chapters/factrevlp.tex
\section{Competitive Analysis via Factor Revealing LP}
\label{sec:factor-revealing}

Fix an optimal offline solution. We will heavily use the structure of this
solution being a~collection of stars, each of them composed of a single open
facility and a set of clients connected in the optimal solution to this
facility. Just as in the analysis of the JMS algorithm~\cite{JaMaSa02,JaMMSV03}, we will focus on 
a~single star~$S$ of \OPT, and compare $\sum_{j \in S} \alpha_j$ to
the cost of this star.

Let $\tau$ be the time of opening the facility $f$ in the considered star $S$ of
the optimal solution. Note that our online algorithm could connect clients $j
\in S$ to facilities opened by the online algorithm both before and after
$\tau$. For a client $j \in S$ that arrived at time $t_j$ and got connected by
the algorithm at time $t_j^c$, we set $a_j = t_j - \tau$ and $s_j = t_j^c -
\tau$ to denote the arrival time and the \emph{service time} of $j$ (time when
$j$ is connected to a facility by the algorithm) relative to~$\tau$. Note that
our algorithm grows $\alpha_j$ until it reaches value $\gamma \cdot (s_j -
a_j)$. Variable~$d_j$ will denote the distance between the locations of the client $j$ and
the facility $f$. Let $\fcg$ denote the cost of opening the facility $f$.

Consider the following factor revealing LP:

\begin{linprog}\label{lp:1}\hypertarget{lpone}{}
\begin{subequations}%
    \label{eq:lp-lowerbound}
    \begin{align}
        z_k{(\gamma)} =
        \max \frac
            {(1+\gamma) \cdot \sum_{i = 0}^{k-1} (s_i - a_i)}
            {\fcg+ \sum_{i = 0}^{k-1} (d_i + |a_i|)}%
        \label{lb:max}
        && \\
        s_i \leq s_{i+1} &&  \forall_{0 \leq i < k-1}%
        \label{lb:time_order}
        \\
        (\gamma-1) \cdot s_i - \gamma \cdot a_i 
            \leq  d_i + d_j +
            (\gamma-1) \cdot s_j - \gamma \cdot a_j && \forall_{0 \leq j < i < k}%
        \label{lb:triangle_ineq}
        \\
        \sum_{i = \ell}^{k-1} \max{\{\gamma \cdot (s_\ell - a_i) - d_i, 0\}} \leq \fcg
        &&\forall_{\ell < k}%
        \label{lb:f_lowerbound}
        \\
        d_i \geq 0, s_i \geq a_i && \forall_{0 \leq i < k}%
        \label{lb:non-negative}
    \end{align}
\end{subequations}
\end{linprog}
\begin{lemma}\label{lem:zk-star-bound}
    $z_k(\gamma)$ is an upper bound on the competitive ratio of the algorithm
    for a fixed value of parameter~$\gamma$ on a star of \OPT with $k$ clients. 
\end{lemma}
\begin{proof}
    It suffices to argue that the
    constraints~(\refeq{lb:time_order}-\refeq{lb:non-negative}) are satisfied for
    any values $d_i, a_i, s_i$ representing the situation of clients from a
    single star of \OPT in an actual run of the online algorithm.

    We consider the set of clients in the order of them being connected by the
    online algorithm. If two clients are connected at the same time, we first
    take the one that arrived first. Constraints~\eqref{lb:time_order} are
    trivially satisfied by this ordering.

    Constraints~\eqref{lb:f_lowerbound} reflect the situation just before time $t_\ell^c$
    when client $\ell$ gets connected by \cref{algo:two-sided}.
    The left-hand side is a sum of the contributions of still active clients
    towards opening a facility in the very same spot as \OPT has
    a facility for this star. Obviously, these contributions cannot exceed the
    cost of opening a facility.

    Finally, to argue for Constraints~\eqref{lb:triangle_ineq} being satisfied
    we consider two cases. If $s_i = s_j$, then $a_j \leq a_i$, and the
    constraint is trivially satisfied. Otherwise, $s_i > s_j$. Consider the
    moment just before client $i$ gets connected. Client $i$ cannot have a budget
    that would be more than sufficient to connect to the facility where client
    $j$ is connected. The connectivity budget of $i$ is then $\gamma \cdot (s_i
    - a_i)$ and the distance to the facility serving $j$ plus waiting at this
    facility for $i$ can be upper bounded by $d_i + d_j+ (s_i - s_j) + \gamma
    \cdot (s_j - a_j)$, see Figure~\ref{fig:triangle}.
\end{proof}
\begin{figure}
    \centering
    \includegraphics[width=10cm]{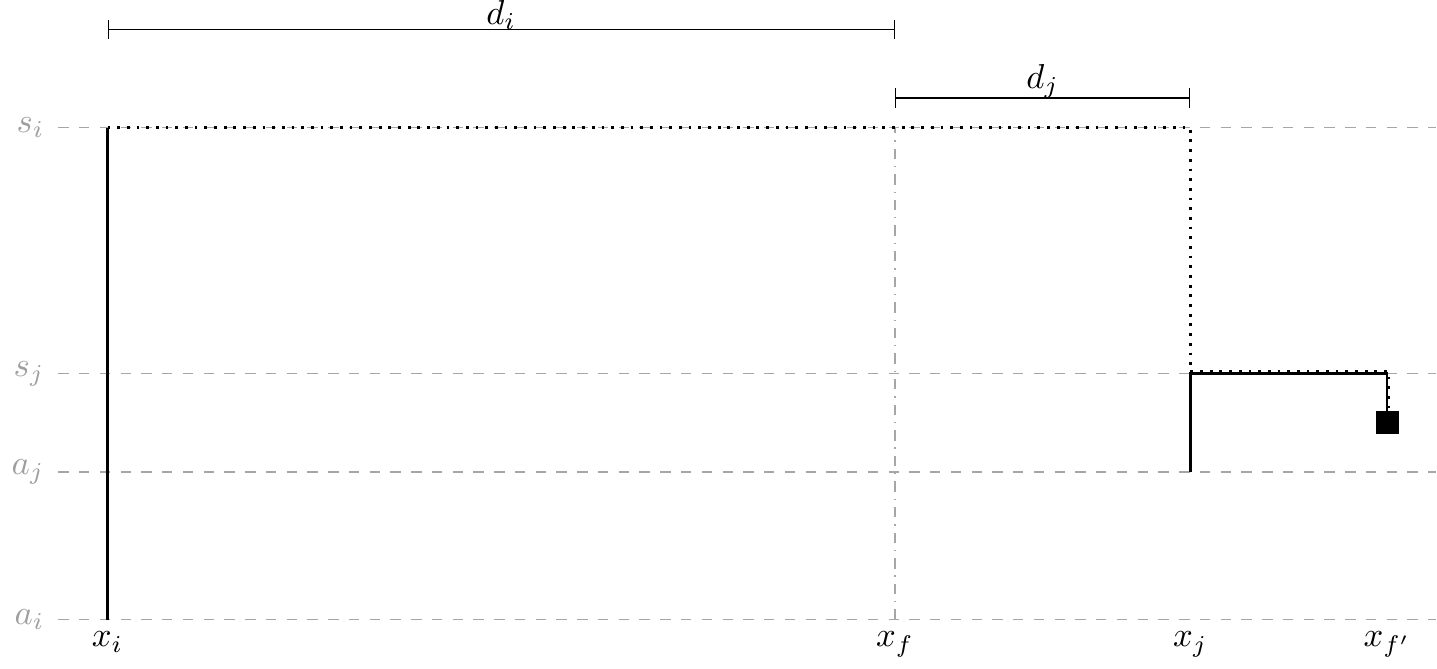}
    \caption{Illustration for Constraint~\eqref{lb:triangle_ineq}
    (the \emph{triangle inequality}). Before time $s_i$, the budget of client
    $i$ must not be sufficient to connect $i$ to the facility $f'$ currently
    serving $j$---otherwise the algorithm would have already connected it. The
    cost of such a connection can be bounded by a sum of $d_i + d_j$ (upper
    bounds $\dist(x_i, x_j)$), $s_i - s_j$ (the waiting cost on the facility side
    since $s_j$), and $\gamma \cdot (s_j - a_j)$ (upper-bounds $\dist(x_j, y_{f'})$
    and the remainder of waiting of $f'$).}
    \label{fig:triangle}
\end{figure}

\subsection{Bounding the LP value}

Here we will show an alternative (an in our opinion conceptually much simpler than the original one from Sec 5.2 of~\cite{JaMMSV03}) method to upper-bound a sequence of factor-revealing linear programs.

Our task now is to bound the value of $z_k(\gamma)$ for every $k$. It is
insufficient to compute $z_k(\gamma)$ for some fixed $k$ as it is monotonically
increasing with $k$ (a fact which we prove in a~slightly weaker form). In
this section, we will however show that it converges to a constant as $k$ goes to
infinity.
\begin{proposition}%
    \label{prop:z-increasing}
    For any $\gamma>0$ and $k, m\in \Nat_+$, it holds that $z_k(\gamma) \leq
    z_{k\cdot m}(\gamma)$.
\end{proposition}
\begin{proof}
    Let $\tuple{d, a, s}\in \Real^k \times \Real^k \times \Real^k$ be a solution
    to~\lpone optimizing $z_k(\gamma)$. We will construct 
    a~solution $\tuple{d', a', s'}\in \Real^{k\cdot m} \times \Real^{k\cdot m}
    \times \Real^{k\cdot m}$ of the same value.

    For $i \in \upto{k}$ and $r \in \upto{m}$, let $d'_{m\cdot i + r} =
    d_i / m$, $a'_{m\cdot i + r} = a_i / m$, and $s'_{m\cdot i + r}
    = s_i / m$. The
    inequalities~\eqref{lb:time_order},~\eqref{lb:triangle_ineq},
    and~\eqref{lb:non-negative} are satisfied by the new solution, since the
    same numbers appear in these inequalities in the solution $\tuple{d, a, s}$.
     To show feasibility of $\tuple{d', a', s'}$ we need to argue
    that~\eqref{lb:f_lowerbound} is satisfied for $\ell$ not divisible by $m$:
    \begin{align*}
            \sum_{i = \ell}^{m\cdot k-1} &
                \max{\{\gamma \cdot (s'_\ell - a'_i) - d'_i, 0\}}  \\
            & \leq
            \sum_{i = m \cdot \fl{\ell / m}}^{m\cdot k-1}
                \max{\{\gamma \cdot (s'_\ell - a'_i) - d'_i, 0\}} 
            && \textit{(more summands)} \\
            &=
            \sum_{i = m \cdot \fl{\ell/m}}^{m\cdot k-1}
                \max{\{\gamma \cdot (s'_{m \cdot \fl{\ell/m}} - a'_i) - d'_i, 0\}} 
            && \textit{(since $s'_\ell = s'_{m \cdot \fl{\frac{\ell}{m}}}$)} \\
            &=
            \sum_{i = \fl{\ell/m}}^{k-1}
                \max{\{\gamma \cdot (s_{\fl{\ell/m}} - a_i) - d_i, 0\}} 
            && \textit{(by definition of $\tuple{d', a', s'}$)} \\
            &\leq \fcg.  &&
            \qedhere
    \end{align*}
\end{proof}
To determine how $z(\gamma)$ converges, we will define another LP, whose optimal
value will always upper-bound~\lpone.

\begin{linprog}\label{lp:2}\hypertarget{lptwo}{}
This program with optimal value equal $y_k(\gamma)$ is defined
to be the same as \lpone except for the following inequality swapped for \eqref{lb:f_lowerbound}:
\begin{subequations}%
    \label{eq:lp-upperbound}
    \begin{align}
        \sum_{i = \ell\mathbf{+1}}^{k-1} \max{\{\gamma \cdot (s_\ell - a_i) - d_i, 0\}} \leq \fcg
        &&\forall_{l < k}.%
        \tag{\ref{eq:lp-upperbound}d}
        \label{lb:f_upperbound}
    \end{align}
\end{subequations}
\end{linprog}%
\begin{proposition}%
    \label{prop:y-decreasing}
    For any $\gamma>0$ and $k, m\in \Nat_+$ it holds that $y_k(\gamma) \geq
    y_{k\cdot m}(\gamma)$.
\end{proposition}
\begin{proof}
    Let $\tuple{d, a, s}\in \Real^{k\cdot m} \times \Real^{k\cdot m} \times
    \Real^{k\cdot m}$ be the solution to~\lptwo maximizing
    $y_{k \cdot m}(\gamma)$. We define a solution $\tuple{d', a', s'}\in \Real^k
    \times \Real^k \times \Real^k$ of the same value as: 
    \begin{align*}
        d'_i = \sum_{r=0}^{m-1} d_{m\cdot i + r},
        &&
        a'_i = \sum_{r=0}^{m-1} a_{m\cdot i + r},
        &&
        s'_i = \sum_{r=0}^{m-1} s_{m\cdot i + r},
        &&\forall_{0 \leq i < k}. 
    \end{align*}
    Again, we only need to focus on the
    inequality~\eqref{lb:f_upperbound}:
    \begin{align*}
            \sum_{i = \ell+1}^{k-1} & \max\set*{\gamma \cdot (s'_\ell - a'_i) - d'_i, 0}  \\
            &=
            \sum_{i = \ell+1}^{k-1} \max\set*{
                \sum_{r=0}^{m-1}
                \gamma \cdot (s_{m\cdot \ell + r} - a_{m\cdot i + r}) - d_{m\cdot i + r}, 0} 
            && \textit{(def.~of $\tuple{d', a', s'}$)} \\
            &\leq
            \sum_{i = \ell+1}^{k-1} \max\set*{
                \sum_{r=0}^{m-1}
                \gamma \cdot (s_{m\cdot (\ell+1) - 1} - a_{m\cdot i+r}) - d_{m\cdot i+r}, 0} 
            && \textit{(monotonicity of $s_i$)} \\
            &\leq
            \sum_{i = \ell+1}^{k-1}
                \sum_{r=0}^{m-1}
                \max\set*{\gamma \cdot (s_{m\cdot (\ell+1) - 1} - a_{m\cdot i+r}) - d_{m\cdot i+r}, 0} 
            && \textit{(Jensen's inequality)} \\
            &=
            \sum_{i = m \cdot (\ell+1)}^{m\cdot k-1}
                \max\set*{\gamma \cdot (s_{m\cdot(\ell+1)-1} - a_i) - d_i, 0} \\
            &\leq \fcg . &&
            \qedhere
      \end{align*}
\end{proof}
Finally we show the relation between the corresponding $z$ and $y$ values.
\begin{proposition}%
    \label{prop:z-y-order}
    For any $\gamma>0$ and $k\in \Nat_+$ it holds that $y_k(\gamma) \geq
    z_k(\gamma)$.
\end{proposition}
\begin{proof}
    Let $\tuple{d, w, s}\in \Real^{k} \times \Real^{k} \times \Real^{k}$ be the
    solution to~\lpone maximizing $z_{k}(\gamma)$. The same
    solution if feasible for~\lptwo, as
    ~\eqref{lb:f_upperbound} is weaker (has strictly fewer summands on
    the left side) than~\eqref{lb:f_lowerbound}.
\end{proof}
\noindent
Proposition~\ref{prop:z-y-order} completes the picture and lets us compare
$z(\gamma)$ with $y(\gamma)$ for any $k$, even with the weak monotonicity we
show in Propositions~\ref{prop:z-increasing} and~\ref{prop:y-decreasing}. 

\begin{corollary}\label{cor:inequalities}
For any $\gamma > 0$ and for any $k_1, k_2\in \Nat_+$, we have
    \[
        z_{k_1}(\gamma)
        \overset{\text{P\ref{prop:z-increasing}}}{\leq}
        z_{k_1 \cdot k_2}(\gamma)
        \overset{\text{P\ref{prop:z-y-order}}}{\leq}
        y_{k_1 \cdot k_2}(\gamma)
        \overset{\text{P\ref{prop:y-decreasing}}}{\leq}
        y_{k_2}(\gamma).
    \]
\end{corollary}
We can now solve $y_k(\gamma)$ for some $k$ and obtain a bound on all
${\{z_\ell(\gamma)\}}_{\ell\in \Nat_+}$. The higher $k$ we choose, the
more precise bound we get in return.

%% file: chapters/computational.tex
\section{Computation of the competitive bounds}\label{sec:computational}


Having \Cref{cor:inequalities} ready to use, we now can lean on a
major strength of our proof strategy: We can now simply use a linear
programming solver to compute a valid bound the competitive ratio.

\begin{figure}
  \centering
    \includegraphics[width=\textwidth]{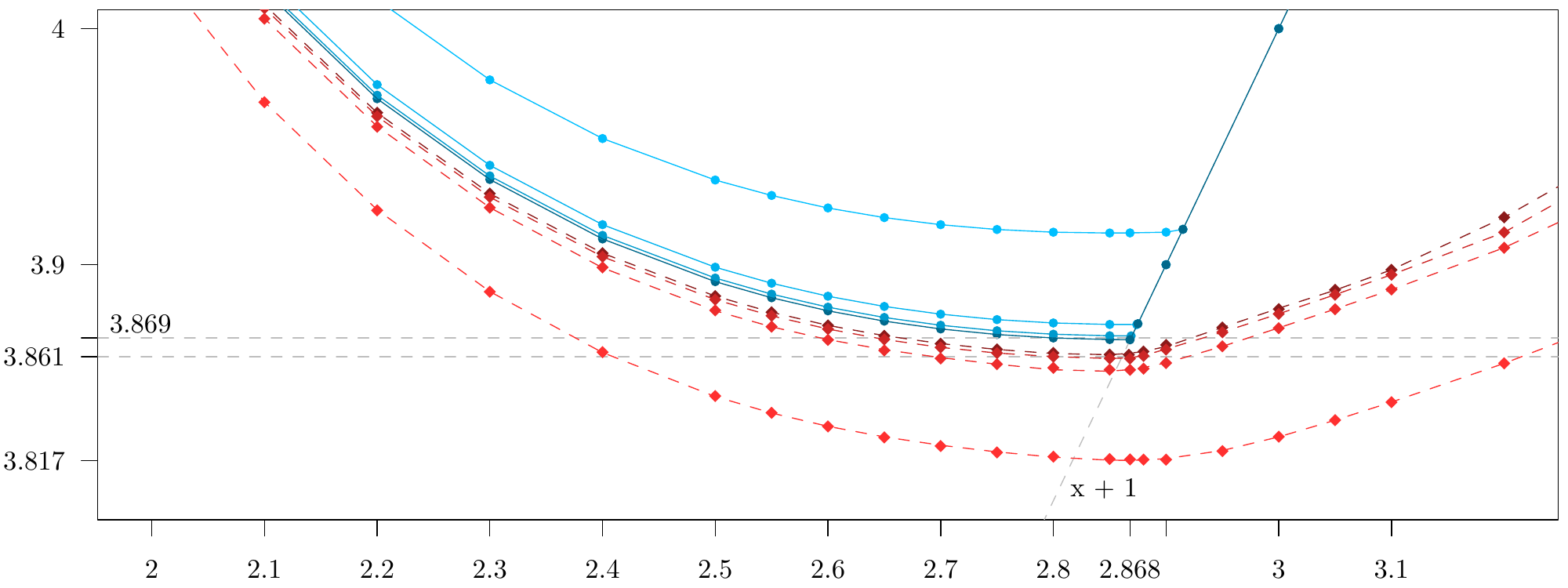}
    \caption{Visual representation of the computed bounds on the competitive ratio as functions of $\gamma$. Each data point corresponds to a single optimal solution of \protect\lpone or \protect\lptwo,
      the curves are interpolated. The blue (solid) curves correspond to the
      upper bound \protect\lptwo for $k=100$, $500$, $1000$ and $1500$ respectively. The
      red (dashed) curves represent the solutions of \protect\lpone for the same steps of $k$.
      We obtain that our algorithm is $3.869$-competitive for $\gamma=2.868$ (\Cref{clm:upperbound}),
      and it is not better than $3.861$-competitive (\Cref{clm:empirical}).
      Note that in line with results of \Cref{sec:factor-revealing}, for a fixed value
      of $\gamma$, the value of $z_k(\gamma)$ increases with increasing $k$ while $y_k(\gamma)$,
      its upper bound, decreases with $k$.}
    \label{fig:curve}
\end{figure}

Our computational results are summarized in \Cref{fig:curve}. We have
been able to solve LPs with up to 1500 clients, which allows us to
make the following two claims, one formally proved and one empirical:

\begin{claim}\label{clm:upperbound}
There exists an optimal solution for the linear program \lptwo
with parameters $k = 1500$ and $\gamma = 2.868$ of objective value $y_{1500}(2.868) \approx 3.869$.
\end{claim}
\begin{proof}
We provide the proof in the form of feasible LP solution and a
feasible dual solution of the same objective value. The solutions are
available along with the rest of our data at~\cite{datarepository}.
\end{proof}

\begin{empiricalclaim}\label{clm:empirical}
The objective function $z_{1500}(\gamma)$ of \lpone, viewed as a function
of $\gamma$, is minimized for $\gamma \approx 2.867$ with objective
value $z_{1500}(2.867) \approx 3.861$.
\end{empiricalclaim}

With \Cref{clm:upperbound}, we can complete the proof of \Cref{thm:two-sided}:

\begin{proof}[Proof of \Cref{thm:two-sided}]
  We partition the cost of \OPT into costs associated with a single facility that
  \OPT opens, forming a star with the facility in the center. \Cref{lem:zk-star-bound} gives us
  that $z_k(\gamma)$ is an upper bound on the competitive ratio of \Cref{algo:two-sided} for a single star,
  and \Cref{cor:inequalities} implies that computing a single optimal solution of the upper bound \lptwo
  is sufficient for the bound on competitive ratio. Finally, \Cref{clm:upperbound} tells us that for
  $\gamma = 2.868$, the competitive ratio on a single star --- and thus, on all stars --- is, after rounding, at most $3.869$.
\end{proof}

The complementary \Cref{clm:empirical} provides a lower bound on the
efficiency of our method, as it suggests that \Cref{algo:two-sided} is
not better than $3.861$-competitive, and thus our analysis is almost
tight.

In \Cref{fig:curve}, we can see that all curves for LP2 begin to follow the line $y = x+1$
once they intersect it, which is not the case for the dashed curves of LP1. The following
observation explains the structure of feasible solutions once the $y = x+1$ line is crossed.

\begin{observation}\label{obs:quirk}
Consider the linear program \lptwo with parameters $k \ge 2$ and $\gamma$,
and let us set $\fcg + \sum_{i=0}^{k-1} (d_i + |a_i|) = 1$. Then, there is a feasible solution
to \lptwo with objective value $\gamma+1$.
\end{observation}

\begin{proof}
Recall that \lptwo is designed for the case the optimal solution opens
a facility at time $\tau$, which we can think of as $0$. For our feasible solution, we set all
distances to be identically zero and we release the first client at
time $a_0 = -1$, with $a_1 = a_2 = \cdots = a_{k-1} = 0$. The cost of
opening a facility is also zero. For the service times of the
algorithm, we set $s_0 = s_1 = \cdots = s_{k-1} = 0$.

We first double check that indeed, setting these values gives us $\fcg +
\sum_{i=0}^{k-1} (d_i + |a_i|) = 1$; we now proceed to check that the
solution is feasible. As all distances and service times are zero, all
constrains except the type \eqref{lb:f_upperbound} are trivially satisfied.

Let us restate the last set of constraints, \eqref{lb:f_upperbound}:
\begin{subequations}%
    \begin{align*}
        \sum_{i = \ell+1}^{k-1} \max{\{\gamma \cdot (s_\ell - a_i) - d_i, 0\}} \leq \fcg
        &&\forall_{\ell < k}.%
    \end{align*}
\end{subequations}

Observe that $a_0$ does not appear in any constraint of this type, and
so it is never checked that $\gamma (s_0 - a_0) \leq \fcg = 0$, which would
be false for any $s_0 > -1$. All constraints of this type only check variables
$a_1, \ldots, a_{k-1}$ and $s_1, \ldots, s_{k-1}$, which are all set to zero,
causing the constraints to be indeed satisfied.
\end{proof}

The quirk from \Cref{obs:quirk} does not cause any issues
for us, as the minimum of $z_k(\gamma)$, which serves as a lower bound
of the competitive ratio of \Cref{algo:two-sided}, is attained before
the curve $z_k(\gamma)$ intersects $y = x+1$ (see \Cref{fig:curve}).

For our computations, we use the LP solver Gurobi Optimizer version
9~\cite{gurobi}. The source code of our generator and selected
few solutions (that can be verified) can be found online
at~\cite{datarepository}. 

%% file: chapters/onesided.tex
\section{From Two-Sided to One-Sided Delay}

We will now show how an online algorithm that solves the
two-sided variant of the facility location problem can be 
transformed into an algorithm for the one-sided variant. 
We require that the former algorithm 
is $(\sensa,\sensb)$-sensible for some constants $\sensa > 1$ and $\sensb > 1$ 
(cf.~\cref{def:sensibility}).
Recall that by \cref{lem:alg1_is_sensible}, 
\cref{algo:two-sided} is $(3/2, 2)$-sensible for $\gamma = 2$.

\subsection{Algorithm definition}

On the basis of an arbitrary online $(\sensa, \sensb)$-sensible algorithm \ALGTS
for the two-sided variant, we construct an online algorithm 
for the one-sided variant in the following way.

\begin{algorithm}
\label{algo:one-sided}
Our algorithm simulates the execution of \ALGTS on the input instance, and 
when \ALGTS opens a facility $f$ at point $y$ at time~$\tau$, and connects a subset
of clients to this facility, our algorithm does the same at time $\tau$.

From this point on, our algorithm tracks, in an online manner, all (late) connections
to facility $f$. Clients that are already connected by \ALGTS but not yet
connected by our algorithm are called \emph{pending for $f$}. At some times
(defined below), our algorithm opens a new facility at $y$ (called a \emph{copy} of $f$),
and connects all clients pending for $f$ to this copy.

\begin{itemize}
\item Let $\tau + b$ be the earliest time when $p \cdot b \geq \fc(y)$, 
where $p$ denotes the number of pending clients. 
(The inequality may be strict if it becomes true because of the increment of $p$.)
If there are no pending clients at time $\tau + \fc(y)$, we set $b = \fc(y)$. 

In either case, the first copy of $f$ is opened at time $\tau + b$. 

\item Let 
\[
    \tn =
    \frac{\sensa \cdot \sensb}{\sensa-1} \cdot \frac{\fc(y)}{b}\,,
\]
Note that $\tn > 1$ by the choice of $b$.
Let $q = \sqrt{\log \tilde{n}}$ and 
let $\ell$ be the smallest integer such that $q^\ell > \tn$. 
Subsequent facility copies are opened at times $\tau + b \cdot q^i$
for all integers $i \in \{1, \ldots, \ell\}$.
\end{itemize}
\end{algorithm}

In total, \cref{algo:one-sided} opens a facility at $y$ at time $\tau$, and 
$\ell+1$ of its copies at times $\tau + b \cdot q^i$, for $i \in \{0, \ldots, \ell\}$.

\subsection{Analysis}

    In the following, we assume that \ALGTS is $(\sensa,\sensb)$-sensible for some
fixed constants $\sensa > 1$ and $\sensb > 1$. We show that
\cref{algo:one-sided} is a valid algorithm (i.e., it eventually connects all
clients), and we upper-bound its total cost in comparison to the cost paid by
\ALGTS. 

We perform the cost comparison for each facility $f$ opened by \ALGTS; in the
following, we use $y$ to denote its location and $\tau$ to denote its opening
time in the solution of \ALGTS. We also use values of $b$, $\tn$, $q$ and~$\ell$
as computed by \cref{algo:one-sided} when handling clients connected to $f$ by
\ALGTS.

\paragraph*{Correctness}

We start with the following helper claim.

\begin{lemma}
\label{lem:tn_bound}
It holds that 
$\tn \leq n \cdot \sensa \cdot \sensb / (\sensa-1)$.
\end{lemma}

\begin{proof}
Let $p'$ be the number of clients \ALGTS connected to $f$ within the interval
$(\tau, \tau+b]$. 
If $p' = 0$, then 
$b = \fc(y)$ and thus
$\tn = (\sensa \cdot \sensb) / (\sensa-1)$. The lemma follows 
trivially as $n \geq 1$. 
Otherwise, $p' > 0$, and then $p' \cdot b \geq \fc(y)$. In this case,
$\tn = (\sensa \cdot \sensb) \cdot (\sensa-1)^{-1} \cdot \fc(y) / b
\leq (\sensa \cdot \sensb) \cdot (\sensa-1)^{-1} \cdot p'$.
As $p' \leq n$, the lemma follows.
\end{proof}

\begin{lemma}
\label{lem:remaining_clients}
\cref{algo:one-sided} connects all clients.
\end{lemma}

\begin{proof}
The last copy of the facility $f$ is opened by \cref{algo:one-sided} 
at time $\tau + b \cdot q^\ell$. 
Thus, it suffices to show that \ALGTS connects no clients 
to $f$ after this time.

Fix any $t > 0$, any integer $i \geq 0$, and consider time interval 
\[ 
    I^t_i = \Big(\tau + \sensa^i \cdot t, \; \tau + \sensa^{i+1} \cdot t\Big].
\]
As \ALGTS is $(\sensa,\sensb)$-sensible, it connects at most 
$\sensb \cdot \fc(y) / (\sensa^i \cdot t)$ clients within interval $I^t_i$.

Note that $\biguplus_{i=0}^\infty I^t_i = (\tau+\sensa^0 \cdot t,\infty) = (\tau+t, \infty)$.
Thus, for an arbitrary $t$, by summing over all $i \geq 0$, the number of clients 
connected after time $\tau + t$ is at most 
\[
    \sum_{i=0}^\infty \frac{\sensb \cdot \fc(y)}{\sensa^{i} \cdot t} = 
    \frac{\sensa \cdot \sensb}{\sensa-1} \cdot \frac{\fc(y)}{t} 
    = \frac{\tn \cdot b}{t} \,.
\]
Hence, 
the number of clients connected after time $\tau + b \cdot q^\ell$ is at most 
$\tn \cdot b / (b \cdot q^\ell) = \tn / q^\ell < 1$. 
As the number of clients is integral, it must be zero.
\end{proof}


\paragraph*{Bounding the waiting time}

Now we focus on bounding the total waiting time of \cref{algo:one-sided}.
Let $C$ be the set of clients connected by \ALGTS to the facility $f$ at $y$. 
We split $C$ into 
four disjoint parts: the clients connected by \ALGTS at time $\tau$, within interval $(\tau,\tau+b)$, at time $\tau+b$, and after $\tau+b$. We denote these parts 
$C_{= \tau}$, $C_{(\tau,\tau+b)}$, 
$C_{=\tau+b}$ and $C_{> \tau+b}$, respectively.
For any client $j \in C$, let $\waitTS_j$ and $\wait_j$ 
denote its waiting cost in the solutions of \ALGTS and \cref{algo:one-sided}, respectively.

\begin{lemma}
\label{lem:waiting_bound_1}
For any client $j \in C \setminus C_{(\tau,\tau+b)}$, it holds that 
$\wait_j \leq q \cdot \waitTS_j$.
\end{lemma}

\begin{proof}
For any client $j$, let $t_j$ be the time of its arrival 
and $t^c_j$ the time \ALGTS connects it to facility $f$. 

In the case $j \in C_{= \tau}$, in both solutions of \ALGTS and \cref{algo:one-sided}, client $j$ waits till~$\tau$,
and thus $\wait_j = \waitTS_j$. The case $j \in C_{=\tau+b}$ is possible only if $t_j = \tau+b$, and then
$\wait_j = 0 \leq q \cdot \waitTS_j$.

It remains to consider the case $j \in C_{> \tau+b}$, i.e., $t^c_j > \tau + b$.
Let $\waitinit_j
= t^c_j - t_j$; this amount represents the inevitable waiting time
that $j$ incurs in both solutions. 
Note that $\waitTS_j - \waitinit_j = t^c_j - \tau$ 
corresponds to the facility-side waiting cost of $j$ in the solution of~\ALGTS.
By \cref{lem:remaining_clients}, 
$t^c_j \leq \tau + b \cdot q^\ell$. Thus, there exists an~integer $i \in \{1, \ldots, \ell\}$, 
such that $\tau+b \cdot q^{i-1} < t^c_j \leq \tau + b \cdot q^{i}$.
\cref{algo:one-sided} connects $j$ at time $\tau + b \cdot q^i$, and therefore
\[
    \wait_j - \waitinit_j = \tau + b \cdot q^i - t^c_j \leq b \cdot q^i = q \cdot \left(b \cdot q^{i-1} \right) < q \cdot (t^c_j - \tau) = q \cdot (\waitTS_j - \waitinit)\,.
\]
The proof is concluded by adding $\waitinit_j$ to both sides.
\end{proof}

\begin{lemma}
\label{lem:waiting_bound_2}
It holds that 
$\sum_{c \in C} \wait(c) \leq \fc(y) + q \cdot \sum_{c \in C} \waitTS(c)$.
\end{lemma}

\begin{proof}
We use the same notions of $t_j$, $t^c_j$, $\wait_j$, $\waitTS_j$ and $\waitinit_j$ as 
in the previous proof.

Fix any client $j \in C_{(\tau, \tau+b)}$. Clearly, $t^c_j > \tau$. 
Furthermore, $j$ is served by 
\cref{algo:one-sided} at time $\tau+b$, and thus
\[
    \sum_{j \in C_{(\tau,\tau+b)}} \left(\wait_j - \waitinit_j\right) 
    = \sum_{j \in C_{(\tau,\tau+b)}} \left(\tau+b - t_j^c \right) 
    \leq |C_{(\tau,\tau+b)}| \cdot b < \fc(y) \,.
\]
The last inequality follows by the definition of $b$ in 
\cref{algo:one-sided}. By adding $\sum_{c \in C_{(\tau,\tau+b)}} \waitinit_j$ to both sides, we obtain 
\[
    \sum_{j \in C_{(\tau,\tau+b)}} \wait_j
    \leq \fc(y) + \sum_{j \in C_{(\tau,\tau+b)}} \waitinit_j 
    \leq \fc(y) + \sum_{j \in C_{(\tau,\tau+b)}} \waitTS_j \,.
\]
By combining the inequality above with \cref{lem:waiting_bound_1} applied to all 
$C \setminus C_{(\tau,\tau+b)}$, we obtain the lemma statement.
\end{proof}%

\paragraph*{Competitive ratio}

Finally, we use our bounds to prove \cref{thm:one-sided}, i.e., show that 
the competitive ratio of \cref{algo:one-sided}.

\begin{proof}[Proof of \cref{thm:one-sided}]
We fix any $(\sensa,\sensb)$-sensible $O(1)$-competitive algorithm \ALGTS for
the two-sided variant. By \cref{lem:alg1_is_sensible}, such algorithm exists for
$\sensa = 3/2$ and $\sensb = 2$.

We fix any facility opened by \ALGTS at location $y$; let $C$ denote the set of
clients that are connected to this facility in the solution of \ALGTS. Below we
show that the total cost pertaining to clients from~$C$ in the solution of
\cref{algo:one-sided} is at most $O(\log n / \log \log n)$ times larger than the
cost pertaining to these clients in the solution of \ALGTS.

The theorem will then follow by summing this relation over all facilities opened by \ALGTS,
and observing that the value of \OPT for the one-sided variant can be only more expensive 
than \OPT for the
two-sided variant (as the two-sided variant is a relaxation of the one-sided
variant).

We relate parts of cost of solution produced by \cref{algo:one-sided} 
to the corresponding costs of \ALGTS.
\begin{itemize}
\item 
The cost of connecting clients from $C$ by \cref{algo:one-sided} is trivially equal to the connection
cost in the solution of \ALGTS. 

\item 
To bound the waiting cost of clients from $C$, we apply 
\cref{lem:waiting_bound_2} obtaining that 
$\sum_{c \in C} \wait(c) \leq \fc(y) + q \cdot \sum_{c \in C} \waitTS(c)$.
As $q = \sqrt{\log \tn} = O(\sqrt{\log n}) = O(\log n / \log \log n)$, 
the total waiting cost of \cref{algo:one-sided} is at most 
$\fc(y) + O(\log n / \log \log n) \cdot \sum_{c \in C} \waitTS(c)$.

\item \cref{algo:one-sided} opens a facility at location $y$ at time $\tau$ and then 
$\ell+1$ of its copies at times $\tau + b \cdot q^i$ for $i \in \{0, \ldots, \ell\}$.
Thus, its overall opening cost is $(\ell+2) \cdot \fc(y)$. 
Recall that $\ell = \ceil{\log \tn / \log q} = O(\log \tn / \log \log \tn)$.
By \cref{lem:tn_bound}, the latter amount is $O(\log n / \log \log n)$.
Thus, the opening cost of \cref{algo:one-sided} is $O(\log n / \log \log n) \cdot \fc(y)$
while that of \ALGTS is $\fc(y)$.
\end{itemize}
The proof follows by adding guarantees of all the cases above.
\end{proof}